\documentclass[a4paper,UKenglish]{lipics}
 
\usepackage{microtype}

\usepackage{algorithm}
\usepackage{algpseudocode}
\usepackage{graphicx}

\renewcommand{\algorithmicrequire}{\textbf{Initialization:}}

\newcommand{\removed}[1]{}

\newcommand{\rem}{\mathbf}

\newcommand{\bbbn}{\mathbb{N}}

\newcommand{\ca}{\mathcal{A}}

\newcommand{\cc}{\mathcal{C}}

\newcommand{\ra}{\rightarrow}
\newcommand{\rsa}{\rightsquigarrow}

\newcommand{\bs}{\backslash}

\newcommand{\tsf}{\textsf}

\definecolor{DarkRed}{RGB}{182,11,1}

\newtheorem{observation}{Observation}
\theoremstyle{theorem}
\newtheorem{proposition}{Proposition}

\title{How Many Cooks Spoil the Soup?\footnote{Supported in part by the School of EEE/CS of the University of Liverpool, NeST initiative, and the EU IP FET-Proactive project \tsf{MULTIPLEX} under contract no 317532.}}
\titlerunning{How Many Cooks Spoil the Soup?} 

\author[1,2]{Othon Michail}
\author[1,2]{Paul G. Spirakis}
\affil[1]{Department of Computer Science, University of Liverpool, Liverpool, UK}
\affil[2]{Computer Technology Institute and Press ``Diophantus'' (CTI), Patras, Greece\\
  \texttt{michailo@cti.gr, P.Spirakis@liverpool.ac.uk}}
\authorrunning{O. Michail and P.\,G. Spirakis} 

\Copyright{Othon Michail and Paul G. Spirakis}

\keywords{coordinator, parallelism, symmetry, symmetry breaking, population protocol, leader election, majority, parity}

\begin{document}

\maketitle

\begin{abstract}
In this work, we study the following basic question: \emph{``How much parallelism does a distributed task permit?''} Our definition of \emph{parallelism} (or \emph{symmetry}) here is not in terms of speed, but in terms of identical \emph{roles} that processes have at the same time in the execution. For example, we may ask: \emph{``Can a given task be solved by a protocol that always has at least two processes in the same role at the same time?''} (i.e., by a protocol that never elects a \emph{unique leader}). We choose to initiate this study in population protocols, a very simple model that not only allows for a straightforward definition of what a role is, but also encloses the challenge of isolating the properties that are due to the protocol from those that are due to the \emph{adversary scheduler}, who controls the interactions between the processes. In particular, we define the role of a process at a given time to be equivalent to the \emph{state} of the process at that time. Moreover, we isolate the symmetry that is due to the protocol (\emph{inherent symmetry}) by focusing on those schedules that maximize symmetry for that protocol and observing how much symmetry breaking the protocol is forced to achieve in order to solve the problem. To allow for such \emph{symmetry maximizing} schedules we consider \emph{parallel schedulers} that in every step may select a whole collection of pairs of nodes (up to a perfect matching) to interact and not just a single pair. Based on these definitions of symmetric computation, we (i) give a \emph{partial characterization} of the set of predicates on input assignments that can be \emph{stably computed with maximum symmetry}, i.e., $\Theta(N_{min})$, where $N_{min}$ is the minimum multiplicity of a state in the initial configuration, and (ii) we turn our attention to the remaining predicates (that have some essentially different properties) and prove a \emph{strong impossibility result} for the \emph{parity} predicate: the inherent symmetry of any protocol that stably computes it is \emph{upper bounded by a constant that depends on the size of the protocol}.     
\end{abstract}

\section{Introduction}
\label{sec:intro}

George Washington said \emph{``My observation on every employment in life is, that, wherever and whenever one person is found adequate to the discharge of a duty by close application thereto, it is worse executed by two persons, and scarcely done at all if three or more are employed therein''} \cite{Sp1855}. The goal of the present paper is to investigate whether the analogue of this observation in simple distributed systems is true. In particular, we ask whether a task that can be solved when a single process has a crucial duty is still solvable when that (and any other) duty is assigned to more than one process. Moreover, we are interested in quantifying the degree of \emph{parallelism} (also called \emph{symmetry} in this paper) that a task is susceptible of.

Leader election is a task of outstanding importance for distributed algorithms. One of the oldest \cite{An80} and probably still one of the most commonly used approaches \cite{Ly96,AW04,AADFP06,KLO10,GH13} for solving a distributed task in a given setting, is to execute a distributed algorithm that manages to elect a unique leader (or \emph{coordinator}) in that setting and then compose this (either sequentially or in parallel) with a second algorithm that can solve the task by assuming the existence of a unique leader. Actually, it is quite typical, that the tasks of electing a leader and successfully setting up the composition enclose the difficulty of solving many other higher-level tasks in the given setting.

Due to its usefulness in solving other distributed tasks, the leader election problem has been extensively studied, in a great variety of distributed settings \cite{Ly96,AW04,GH13,FSW14,AG15}. Still, there is an important point that is much less understood, concerning whether an election step is \emph{necessary} for a given task and \emph{to what extent} it can be avoided. Even if a task $T$ can be solved in a given setting by first passing through a configuration with a unique leader, it is still valuable to know whether there is a correct algorithm for $T$ that avoids this. In particular, such an algorithm succeeds without the need to ever have less than $k$ processes in a given ``role'', and we are also interested in how large $k$ can be without sacrificing solvability.

Depending on the application, there are several ways of defining what the ``role'' of a process at a given time in the execution is. In the typical approach of electing a unique leader, a process has the leader role if a $leader$ variable in its local memory is set to true and it does not have it otherwise. In other cases, the role of a process could be defined as its complete local history. In such cases, we would consider that two processes have the same role after $t$ steps iff both have the same local history after each  one of them has completed $t$ local steps. It could also be defined in terms of the external interface of a process, for example, by the messages that the process transmits, or it could even correspond to the branch of the program that the process executes. In this paper, as we shall see, we will define the role of a process at a given time in the execution, as the entire content of its local memory. So, in this paper, two processes $u$ and $v$ will be regarded to have the same role at a given time $t$ iff, at that time, the local state of $u$ is equal to the local state of $v$. 

Understanding the parallelism that a distributed task allows, is of fundamental importance for the following reasons. First of all, usually, the more parallelism a task allows, the more efficiently it can be solved. Moreover, the less symmetry a solution for a given problem has to achieve in order to succeed, the more vulnerable it is to faults. For an extreme example, if a distributed algorithm elects in every execution a unique leader in order to solve a problem, then a single crash failure (of the leader) can be fatal.   

\subsection{Our Approach}
\label{subsec:approach}

We have chosen to initiate the study of the above problem in a very minimal distributed setting, namely in Population Protocols of Angluin \emph{et al.} \cite{AADFP06} (see Section \ref{subsec:pw} for more details and references). One reason that makes population protocols convenient for the problem under consideration, is that the role of a process at a given step in the execution can be defined in a straightforward way as the state of the process at the beginning of that step. So, for example, if we are interested in an execution of a protocol that stabilizes to the correct answer without ever electing a unique leader, what we actually require is an execution that, up to stability, never goes through a configuration in which a state $q$ is the state of a single node, which implies that, in every configuration of the execution, every state $q$ is either absent or the state of at least two nodes. Then, it is straightforward to generalize this to any symmetry requirement $k$, by requiring that, in every configuration, every state $q$ is either absent or the state of at least $k$ nodes.

What is not straightforward in this model (and in any model with adversarially determined events), is how to isolate the symmetry that is \emph{only} due to the protocol. For if we require the above condition on executions to be satisfied \emph{for every} execution of a protocol, then most protocols will fail trivially, because of the power of the adversary scheduler. In particular, there is almost always a way for the scheduler to force the protocol to break symmetry maximally, for example, to make it reach a configuration in which some state is the state of a single node, even when the protocol does not have an \emph{inherent} mechanism of electing a unique state. Moreover, though for computability questions it is sufficient to assume that the scheduler selects in every step a single pair of nodes to interact with each other, this type of a scheduler is problematic for estimating the symmetry of protocols. The reason is that even fundamentally parallel operations, necessarily pass through a highly-symmetry-breaking step. For example, consider the rule $(a,a)\ra(b,b)$ and assume that an even number of nodes are initially in state $a$. The goal is here for the protocol to convert all $a$s to $b$s. If the scheduler could pick a perfect matching between the $a$s, then in one step all $a$s would be converted to $b$s, and additionally the protocol would never pass trough a configuration in which a state is the state of fewer than $n$ nodes. Now, observe that the sequential scheduler can only pick a single pair of nodes in each step, so in the very first step it yields a configuration in which state $b$ is the state of only 2 nodes. Of course, there are turnarounds to this, for example by taking into account only equal-interaction configurations, consisting of the states of the processes after all processes have participated in an equal number of interactions, still we shall follow an alternative approach that simplifies the arguments and the analysis.             

In particular, we will consider schedulers that can be maximally parallel. Such a scheduler, selects in every step a matching (of any possible size) of the complete interaction graph, so, in one extreme, it is still allowed to select only one interaction but, in the other extreme, it may also select a perfect matching in a single step. Observe that this scheduler is different both from the sequential scheduler traditionally used in the area of population protocols and from the fully parallel scheduler which assumes that $\Theta(n)$ interactions occur in parallel in every step. Actually, several recent papers \cite{CCDS14,AG15,DS15} assume a fully parallel scheduler implicitly, by defining the model in terms of the sequential scheduler and then performing their analysis in terms of parallel time, defined as the sequential time divided by $n$.   

Finally, in order to isolate the \emph{inherent} symmetry, i.e., the symmetry that is only due to the protocol, we shall focus on those schedules \footnote{By ``schedule'' we mean an ``execution'' throughout.} that achieve as high symmetry as possible for the given protocol. Such schedules may look into the protocol and exploit its structure so that the chosen interactions maximize parallelism. It is crucial to notice that this restriction does by no means affect correctness. Our protocols are still, as usual, required to stabilize to the correct answer in \emph{any} fair execution (and, actually, in this paper against a more generic scheduler than the one traditionally assumed). The above restriction is only a convention for estimating the \emph{inherent} symmetry of a protocol designed to operate in an adversarial setting. On the other hand, one does not expect this \emph{measure of inherent symmetry} to be achieved by the majority of executions. If, instead, one is interested in some \emph{measure of the observed symmetry}, then it would make more sense to study an \emph{expected observed symmetry} under some probabilistic assumption for the scheduler. We leave this as an interesting direction for future research (see Section \ref{sec:conclusions} for more details on this). 

For a given initial configuration, we shall estimate the symmetry breaking performed by the protocol not in any possible execution but an execution in which the scheduler tries to maximize the symmetry. In particular, we shall define the symmetry of a protocol on a given initial configuration $c_0$ as the maximum symmetry achieved over all possible executions on $c_0$. So, in order to lower bound by $k$ the symmetry of a protocol on a given $c_0$, it will be sufficient to present \emph{a} schedule in which the protocol stabilizes without ever ``electing'' fewer than $k$ nodes. On the other hand, to establish an upper bound of $h$ on symmetry, we will have to show that \emph{in every} schedule (on the given $c_0$) the protocol ``elects'' at most $h$ nodes. Then we may define the symmetry of the protocol on a set of initial configurations as the minimum of its symmetries over those initial configurations. The symmetry of a protocol (as a whole) shall be defined as a function of some parameter of the initial configuration and is deferred to Section \ref{sec:preliminaries}.
\begin{observation}
The above definition leads to very strong impossibility results, as these upper bounds are also upper bounds on the observed symmetry. In particular, if we establish that the symmetry of a protocol $\ca$ is at most $h$ then, it is clear that under \emph{any scheduler} the symmetry of $\ca$ is at most $h$. 
\end{observation}

Section \ref{sec:preliminaries} brings together all definitions and basic facts that are used throughout the paper. In Section \ref{sec:symmetric-comp}, we give a set of positive results. The main result here is a partial characterization, showing that a wide subclass of semilinear predicates is computed with symmetry $\Theta(N_{min})$, which is asymptotically optimal. Then, in Section \ref{sec:symmetry-breaking}, we study some basic predicates that seem to require much symmetry breaking. In particular, we study the \emph{majority} and the \emph{parity} predicates. For majority we establish a constant symmetry, while for parity we prove a strong impossibility result, stating that the symmetry of any protocol that stably computes it, is upper bounded by an integer depending only on the size of the protocol (i.e., a constant, compared to the size of the system). The latter implies that there exist predicates which can \emph{only} be computed by protocols that perform some sort of leader-election (not necessarily a unique leader but at most a constant number of nodes in a distinguished leader role). In Section \ref{sec:conclusions}, we give further research directions that are opened by our work. Finally, the Appendices \ref{app:proofs-output-stable}, \ref{app:majority-weak-result}, and \ref{app:experiments} provide all omitted details and proofs.

\subsection{Further Related Work}
\label{subsec:pw}

In contrast to static systems with unique identifiers (IDs) and dynamic systems, the role of symmetry in \emph{static anonymous systems} has been deeply investigated \cite{An80,YK96,Kr97,FMS98}. \emph{Similarity} as a way to compare and contrast different models of concurrent programming has been defined and studied in \cite{JS85}. One (restricted) type of symmetry that has been recently studied in systems with IDs is the existence of \emph{homonyms}, i.e., processes that are initially assigned the same ID \cite{DFGKR11}. Moreover, there are several standard models of distributed computing that do not suffer from a necessity to break symmetry globally (e.g., to elect a leader) like Shared Memory with Atomic Snapshots \cite{AADGMS93,AW04}, Quorums \cite{Sk82,PW95,MRWW01}, and the LOCAL model \cite{Pe00, Su13}. 

Population Protocols  were originally motivated by highly dynamic networks of simple sensor nodes that cannot control their mobility. The first papers focused on the computational capabilities of the model which have now been almost completely characterized. In particular, if the interaction network is complete (as is also the case in the present paper), i.e., one in which every pair of processes may interact, then the computational power of the model is equal to the class of the \emph{semilinear predicates} (and the same holds for several variations) \cite{AAER07}. Interestingly, the generic protocol of \cite{AADFP06} that computes all semilinear predicates, elects a unique leader in every execution and the same is true for the construction in \cite{CDS14}. Moreover, according to \cite{AG15}, all known generic constructions of semilinear predicates ``fundamentally rely on the election of a single initial \emph{leader} node, which coordinates phases of computation''. Semilinearity of population protocols persists up to $o(\log\log n)$ local space but not more than this \cite{MNPS11}. If additionally the connections between processes can hold a state from a finite domain, then the computational power dramatically increases to the commutative subclass of $\rem{NSPACE}(n^2)$ \cite{MCS11-2}. Recently, Doty \cite{Do14} demonstrated the formal equivalence of population protocols to \emph{chemical reaction networks} (CRNs), which model chemistry in a \emph{well-mixed solution}. Moreover, the recently proposed \emph{Network Constructors} extension of population protocols \cite{MS16} is capable of constructing arbitrarily complex stable networks. Czyzowicz \emph{et al.} \cite{CGKK15} have recently studied the relation of population protocols to antagonism of species, with dynamics modeled by discrete Lotka-Volterra equations. Finally, in \cite{CCDS14}, the authors highlighted the importance of executions that necessarily pass through a ``bottleneck'' transition (meaning a transition between two states that have only constant counts in the population, which requires $\Omega(n^2)$ expected number of steps to occur), by proving that protocols that avoid such transitions can only compute existence predicates. To the best of our knowledge, our type of approach, of computing predicates stably without \emph{ever} electing a unique leader, has not been followed before in this area (according to \cite{AG15}, ``\cite{DH15} proposes a leader-less framework for population computation'', but this should not be confused with what we do in this paper, as it only concerns the achievement of dropping the requirement for a \emph{pre-elected} unique leader that was assumed in all previous results for that problem). For introductory texts to population protocols, the interested reader is encouraged to consult \cite{AR09} and \cite{MCS11}.

\section{Preliminaries}
\label{sec:preliminaries}

A \emph{population protocol} (PP) is a 6-tuple $(X,Y,Q,I,O,\delta)$, where $X$, $Y$, and $Q$ are all finite sets and $X$ is the \emph{input alphabet}, $Y$ is the \emph{output alphabet}, $Q$ is the set of \emph{states}, $I\colon X\to Q$ is the \emph{input function}, $O\colon Q\to Y$ is the \emph{output function}, and $\delta\colon Q\times Q\to Q\times Q$ is the \emph{transition function}.

If $\delta(a,b)=(a^{\prime},b^{\prime})$, we call $(a,b)\rightarrow (a^{\prime},b^{\prime})$ a \emph{transition}. A transition $(a,b) \rightarrow (a^{\prime},b^{\prime})$ is called \emph{effective} if $x\neq x^\prime$ for at least one $x\in\{a,b\}$ and \emph{ineffective} otherwise. When we present the transition function of a protocol we only present the effective transitions. The system consists of a population $V$ of $n$ distributed \emph{processes} (also called \emph{nodes}). In the generic case, there is an underlying \emph{interaction graph} $G=(V,E)$ specifying the permissible interactions between the nodes. Interactions in this model are always pairwise. In this work, $G$ is a \emph{complete directed interaction graph}.

Let $Q$ be the set of states of a population protocol $\ca$. A configuration $c$ of $\ca$ on $n$ nodes is an element of $\bbbn_{\geq 0}^{|Q|}$, such that, for all $q\in Q$, $c[q]$ is equal to the number of nodes that are in state $q$ in configuration $c$ and it holds that $\sum_{q\in Q} c[q]=n$. For example, if $Q=\{q_0,q_1,q_2,q_3\}$ and $c=(7,12,52,0)$, then, in $c$, 7 nodes of the $7+12+52+0=71$ in total, are in state $q_0$, 12 nodes in state $q_1$, and 52 nodes in state $q_2$.

Execution of the protocol proceeds in discrete steps and it is determined by an \emph{adversary scheduler} who is allowed to be \emph{parallel}, meaning that, in every step, it may select one or more pairwise interactions (up to a maximum matching) to occur at the same time. This is an important difference from classical population protocols where the scheduler could only select a single interaction per step. More formally, in every step, a non-empty matching $(u_1,v_1),(u_2,v_2),\ldots,(u_k,v_k)$ from $E$ is selected by the scheduler and, for all $1\leq i\leq k$, the nodes $u_i,v_i$ interact with each other and update their states according to the transition function $\delta$. A \emph{fairness condition} is imposed on the adversary to ensure the protocol makes progress. An infinite execution is \emph{fair} if for every pair of configurations $c$ and $c^{\prime}$ such that $c\rightarrow c^{\prime}$ (i.e., $c$ can go in one step to $c^\prime$), if $c$ occurs infinitely often in the execution then so does $c^{\prime}$.

In population protocols, we are typically interested in computing predicates on the inputs, e.g., $N_a\geq 5$, being true whenever there are at least 5 $a$s in the input. \footnote{We shall use throughout the paper $N_i$ to denote the number of nodes with input/state $i$.} Moreover, computations are \emph{stabilizing} and not terminating, meaning that it suffices for the nodes to eventually converge to the correct output. We say that a protocol \emph{stably computes} a predicate if, on any population size, any input assignment, and any fair execution on these, all nodes eventually stabilize their outputs to the value of the predicate on that input assignment.

We define the \emph{symmetry $s(c)$ of a configuration $c$} as the \emph{minimum multiplicity of a state that is present in $c$} (unless otherwise stated, in what follows by ``symmetry'' we shall always mean ``inherent symmetry''). That is, $s(c)=\min_{q\in Q \;:\; c[q]\geq 1}\{c[q]\}$. For example, if $c=(0,4,12,0,52)$ then $s(c)=4$, if $c=(1,\ldots)$ then $s(c)=1$, which is the minimum possible value for symmetry, and if $c=(n,0,0,\ldots,0)$ then $s(c)=n$ which is the maximum possible value for symmetry. So, the range of the symmetry of a configuration is $\{1,2,\ldots,n\}$.

Let $\mathcal{C}_0(\ca)$ be the set of all \emph{initial configurations} for a given protocol $\ca$. Given an initial configuration $c_0\in \mathcal{C}_0(\ca)$, denote by $\Gamma(c_0)$ the set of all fair executions of $\ca$ that begin from $c_0$, each execution being truncated to its prefix \emph{up to stability}. \footnote{In this work, we only require protocols to preserve their symmetry \emph{up to stability}. This means that a protocol is allowed to break symmetry arbitrarily after stability, e.g., even elect a unique leader, without having to pay for it. We leave as an interesting open problem the comparison of this convention to the apparently harder requirement of maintaining symmetry forever.} 

Given any initial configuration $c_0$ and any execution $\alpha\in \Gamma(c_0)$, define the \emph{symmetry breaking of $\ca$ on $\alpha$} as the difference between the symmetry of the initial configuration of $\alpha$ and the minimum symmetry of a configuration of $\alpha$, that is, the \emph{maximum drop in symmetry} during the execution. Formally, $b(\ca,\alpha)=s(c_0)-\min_{c\in \alpha}\{s(c)\}$. Also define the \emph{symmetry of $\ca$ on $\alpha$} as $s(\ca,\alpha)=\min_{c\in \alpha}\{s(c)\}$. Of course, it holds that $s(\ca,\alpha)=s(c_0)-b(\ca,\alpha)$. Moreover, observe that, for all $\alpha\in \Gamma(c_0)$, $0\leq b(\ca,\alpha)\leq s(c_0)-1$ and $1\leq s(\ca,\alpha)\leq s(c_0)$. In several cases we shall denote $s(c_0)$ by $N_{min}$.

The \emph{symmetry breaking of a protocol $\ca$ on an initial configuration $c_0$} can now be defined as $b(\ca,c_0)=\min_{\alpha\in \Gamma(c_0)}\{b(\ca,\alpha)\}$ and:
\begin{definition}
We define the \emph{symmetry of $\ca$ on $c_0$} as $s(\ca,c_0)=\max_{\alpha\in \Gamma(c_0)}\{s(\ca,\alpha)\}$.
\end{definition}

\begin{remark}
To estimate the \emph{inherent} symmetry with which a protocol computes a predicate on a $c_0$, we execute the protocol against an \emph{imaginary} scheduler who is a \emph{symmetry maximizer}.
\end{remark}

Now, given the set $\cc(N_{min})$ of all initial configurations $c_0$ such that $s(c_0)=N_{min}$, we define the \emph{symmetry breaking of a protocol $\ca$ on $\cc(N_{min})$} as $b(\ca,N_{min})=\max_{c_0\in \cc(N_{min})}\{b(\ca,c_0)\}$ and: 
\begin{definition}
We define the \emph{symmetry of $\ca$ on $\cc(N_{min})$} as $s(\ca,N_{min})=\min_{c_0\in \cc(N_{min})}\{s(\ca,c_0)\}$.
\end{definition}
Observe again that $s(\ca,N_{min})=N_{min}-b(\ca,N_{min})$ and that $0\leq b(\ca,N_{min})\leq N_{min}-1$ and $1\leq s(\ca,N_{min})\leq N_{min}$.

This means that, in order to establish that a protocol $\ca$ is at least $g(N_{min})$ symmetric asymptotically (e.g., for $g(N_{min})=\Theta(\log N_{min})$), we have to show that for every sufficiently large $N_{min}$, the symmetry breaking of $\ca$ on $\cc(N_{min})$ is at most $N_{min}-g(N_{min})$, that is, to show that for all initial configurations $c_0\in\cc(N_{min})$ there exists \emph{an} execution on $c_0$ that drops the initial symmetry by at most $N_{min}-g(N_{min})$, e.g., by at most $N_{min}-\log N_{min}$ for $g(N_{min})=\log N_{min}$, or that does not break symmetry at all in case $g(N_{min})=N_{min}$. On the other hand, to establish that the symmetry is at most $g(N_{min})$, e.g., at most $1$ which is the minimum possible value, one has to show a symmetry breaking of at least $N_{min}-g(N_{min})$ on infinitely many $N_{min}$s.

\section{Predicates of High Symmetry}
\label{sec:symmetric-comp}

In this section, we try to identify predicates that can be stably computed with much symmetry. We first give an indicative example, then we generalize to arrive at a partial characterization of the predicates that can be computed with maximum symmetry, and, finally, we highlight the role of output-stable states in symmetric computations.   

\subsection{An Example: Count-to-$x$}
\label{subsec:count-to-x}

\noindent\textbf{Protocol} \emph{Count-to-$x$}: $X=\{0,1\}$, $Q=\{q_0,q_1,q_2,\ldots,q_x\}$, $I(\sigma)=q_{\sigma}$, for all $\sigma\in X$, $O(q_x)=1$ and $O(q)=0$, for all $q\in Q\bs\{q_x\}$, and $\delta$: $(q_i,q_j)\ra (q_{i+j},q_0)$, if $i+j<x$, $(q_i,q_j)\ra (q_x,q_x)$, otherwise.

\begin{proposition} \label{pro:count-to-x}
The symmetry of Protocol \emph{Count-to-$x$}, for any $x=O(1)$, is at least $(2/3)\lfloor N_{min}/x\rfloor - (x-1)/3$, when $x\geq 2$, and $N_{min}$, when $x=1$; i.e., it is $\Theta(N_{min})$ for any $x=O(1)$.
\end{proposition}
\begin{proof}
The scheduler \footnote{Always meaning the \emph{imaginary symmetry-maximizing scheduler} when lower-bounding the symmetry.} partitions the $q_1$s, let them be $N_1(0)$ initially and denoted just $N_1$ in the sequel, into $\lfloor N_1/x\rfloor$ groups of $x$ $q_1$s each, possibly leaving an incomplete group of $r\leq x-1$ $q_1$s residue. Then, in each complete group, it performs a sequential gathering of $x-3$ other $q_1$s to one of the nodes, which will go through the states $q_1,q_2,\ldots,q_{x-1}$. The same gathering is performed in parallel to all groups, so every state that exists in one group will also exist in every other group, thus, its cardinality never drops below $\lfloor N_1/x\rfloor$. In the end, at step $t$, there are many $q_0$s, $N_{x-1}(t)=\lfloor N_1/x\rfloor$, and $N_1(t)=\lfloor N_1/x\rfloor + r$, where $0\leq r\leq x-1$ is the residue of $q_1$s. That is, in all configurations so far, the symmetry has not dropped below $\lfloor N_1/x\rfloor$. 

Now, we cannot pick, as a symmetry maximizing choice of the scheduler, a perfect bipartite matching between the $q_1$s and the $q_{x-1}$s converting them all to the alarm state $q_x$, because this could possibly leave the symmetry-breaking residue of $q_1$s. What we can do instead, is to match in one step as many as we can so that, after the corresponding transitions, $N_x(t^\prime)\geq N_1(t^\prime)$ is satisfied. In particular, if we match $y$ of the $(q_1,q_{x-1})$ pairs we will obtain $N_x(t^\prime)=2y$, $N_{x-1}(t^\prime)=\lfloor N_1/x\rfloor-y$, and $N_1(t^\prime)=\lfloor N_1/x\rfloor-y+r$ and what we want is
\begin{equation*}
2y \geq \lfloor N_1/x\rfloor - y + r \Rightarrow 3y \geq \lfloor N_1/x\rfloor + r \Rightarrow y \geq \frac{\lfloor N_1/x\rfloor + r}{3},
\end{equation*}  
which means that if we match approximately $1/3$ of the $(q_1,q_{x-1})$ pairs then we will have as many $q_x$ as we need in order to eliminate all $q_1$s in one step and all remaining $q_{x-1}$s in another step.

The minimum symmetry in the whole course of this schedule is
\begin{align*}
N_{x-1}(t^\prime) &= \lfloor N_1/x\rfloor - y
= \lfloor N_1/x\rfloor - \frac{\lfloor N_1/x\rfloor + r}{3}\\
&= \frac{2}{3}\lfloor N_1/x\rfloor - \frac{r}{3}
\geq \frac{2}{3}\lfloor N_1/x\rfloor - \frac{x-1}{3}.
\end{align*}

So, we have shown that if there are no $q_0$s in the initial configuration, then the symmetry breaking of the protocol on the schedule defined above is at most $N_{min}-((2/3)\lfloor N_1/x\rfloor - (x-1)/3)=N_{min}-((2/3)\lfloor N_{min}/x\rfloor - (x-1)/3)$. Next, we consider the case in which there are some $q_0$s in the initial configuration. Observe that in this protocol the $q_0$s can only increase, so their minimum cardinality is precisely their initial cardinality $N_0$. Consequently, in case $N_0\geq 1$ and $N_1\geq 1$, and if $N_{min}=\min\{N_0,N_1\}$, the symmetry breaking of the schedule defined above is $N_{min}-\min\{N_0,N_{x-1}(t^\prime)\}$. If, for some initial configuration, $N_0 \geq N_{x-1}(t^\prime)$ then the symmetry breaking is $N_{min}-N_{x-1}(t^\prime)\leq N_{min}-((2/3)\lfloor N_1/x\rfloor - (x-1)/3)$. This gives again $N_{min}-((2/3)\lfloor N_{min}/x\rfloor - (x-1)/3)$, when $N_1\leq N_0$, and less than $N_{min}-((2/3)\lfloor N_{min}/x\rfloor - (x-1)/3)$, when $N_1>N_0=N_{min}$. If instead, $N_0 < N_{x-1}(t^\prime) < N_1$, then, in this case, the symmetry breaking is $N_{min}-\min\{N_0,N_{x-1}(t^\prime)\}=N_0-N_0=0$. Finally, if $N_0=n$, then the symmetry breaking is 0. We conclude that for every initial configuration, the symmetry breaking of the above schedule is at most $N_{min}-N_{x-1}(t^\prime)\leq N_{min}-((2/3)\lfloor N_{min}/x\rfloor - (x-1)/3)$, for all $x\geq 2$, and 0, for $x=1$. Therefore, the symmetry of the \emph{Count-to-$x$} protocol is at least $(2/3)\lfloor N_{min}/x\rfloor + (x-1)/3=\Theta(N_{min})$, for $x\geq 2$, and $N_{min}$, for $x=1$.
\end{proof}

\subsection{A General Positive Result}
\label{subsec:general-positive}

\begin{theorem} \label{the:general-positive}
Any predicate of the form $\sum_{i\in [k]} a_iN_i\geq c$, for integer constants $k\geq 1$, $a_i\geq 1$, and $c\geq 0$, can be computed with symmetry more than $\lfloor N_{min}/(c/\sum_{j\in L} a_j+2)\rfloor-2=\Theta(N_{min})$.
\end{theorem}
\begin{proof}
We begin by giving a parameterized protocol (Protocol \ref{prot:positive-lc}) that stably computes any such predicate, and then we shall prove that the symmetry of this protocol is the desired one.

\floatname{algorithm}{Protocol}
\renewcommand{\algorithmiccomment}[1]{// #1}
\begin{algorithm}[!h]
  \caption{\emph{Positive-Linear-Combination}}\label{prot:positive-lc}
  \begin{algorithmic}
    \medskip
    \State $Q=\{q_0,q_1,q_2,\ldots,q_c\}$
    \State $I(\sigma_i)=q_{a_i}$, for all $\sigma_i\in X$
    \State $O(q_c)=1$ and $O(q)=0$, for all $q\in Q\bs\{q_c\}$
    \State $\delta$: 
    \begin{align*}
    (q_i,q_j)&\ra (q_{i+j},q_0), \mbox{ if } i+j<c\\
    &\ra (q_c,q_c), \mbox{ otherwise}
    \end{align*}
  \end{algorithmic}
\end{algorithm}

Take now any initial configuration $C_0$ on $n$ nodes and let $L\subseteq [k]$ be the set of indices of the initial states that are present in $C_0$. Let also $q_{min}$ be the state with minimum cardinality, $N_{min}$, in $C_0$. Construct $\lfloor N_{min}/x\rfloor$ groups, by adding to each group $x=\lceil c/\sum_{j\in L} a_j\rceil$ copies of each initial state. Observe that each group has total sum $\sum_{j\in L} a_jx=x\sum_{j\in L} a_j=\lceil c/\sum_{j\in L} a_j\rceil(\sum_{j\in L} a_j)\geq c$. Moreover, state $q_{min}$ has a residue $r_{min}$ of at most $x$ and every other state $q_i$ has a residue $r_i\geq r_{min}$. Finally, keep $y= \lceil (N_{min}+r_{min})/(x+1)\rceil-1$ from those groups and drop the other $\lfloor N_{min}/x\rfloor-y$ groups making their nodes part of the residue, which results in new residue values $r^\prime_j=x(\lfloor N_{min}/x\rfloor-y)+r_j$, for all $j\in L$. It is not hard to show that $y\leq r^\prime_j$, for all $j\in L$.  

We now present a schedule that achieves the desired symmetry. The schedule consists of two phases, the \emph{gathering} phase and the \emph{dissemination} phase. In the dissemination phase, the schedule picks a node of the same state from every group and starts aggregating to that node the sum of its group sequentially, performing the same in parallel in all groups. It does this until the alarm state $q_c$ first appears. When this occurs, the dissemination phase begins. In the dissemination phase, the schedule picks one after the other all states that have not yet been converted to $q_c$. For each such state $q_i$, it picks a $q_c$ which infects one after the other (sequentially) the $q_i$s, until $N_c(t)\geq N_i(t)$ is satisfied for the first time. Then, in a single step that matches each $q_i$ to a $q_c$, it converts all remaining $q_i$s to $q_c$.

We now analyze the symmetry breaking of the protocol in this schedule. Clearly, the initial symmetry is $N_{min}$. As long as a state appears in the groups, its cardinality is at least $y$, because it must appear in each one of them. When a state $q_i$ first becomes eliminated from the groups, its cardinality is equal to its residue $r_i^\prime$. Thus, so far, the minimum cardinality of a state is 
\begin{equation*}
\min\{y,\min_{j\in L} r^\prime_j\}=y=\left\lceil \frac{N_{min}+r_{min}}{x+1}\right\rceil-1>\left\lfloor \frac{N_{min}}{c/\sum_{j\in L} a_j+2}\right\rfloor-2.
\end{equation*}
It follows that the maximum symmetry breaking so far is less than
$N_{min}-\left\lfloor \frac{N_{min}}{c/\sum_{j\in L} a_j+2}\right\rfloor+2$.

Finally, we must also take into account the dissemination phase. In this phase, the $q_c$s are $2y$ initially and can only increase, by infecting other states, until they become $n$ and the cardinalities of all other states decrease until they all become 0. Take any state $q_i\neq q_c$ with cardinality $N_i(t)$ when the dissemination phase begins. What the schedule does is to decrement $N_i(t)$, until $N_c(t^\prime)\geq N_i(t^\prime)$ is first satisfied, and then to eliminate all occurrences of $q_i$ in one step. Due to the fact that $N_i$ is decremented by one in each step resulting in a corresponding increase by one of $N_c$, when $N_c(t^\prime)\geq N_i(t^\prime)$ is first satisfied, it holds that $N_i(t^\prime)\geq N_c(t^\prime)-1\geq N_c(t)-1\geq 2y-1\geq y$ for all $y\geq 1$, which implies that the lower bound of $y$ on the minimum cardinality, established for the gathering phase, is not violated during the dissemination phase. 

We conclude that the symmetry of the protocol in the above schedule is more than $\lfloor N_{min}/(c/\sum_{j\in L} a_j+2)\rfloor-2$. 
\end{proof}

\subsection{Output-stable States}
\label{subsec:output-stable}

Informally, a state $q\in Q$ is called \emph{output-stable} if its appearance in an execution guarantees that the output value $O(q)$ must be the output value of the execution. More formally, if $q$ is output-stable and $C$ is a configuration containing $q$, then the set of outputs of $C^\prime$ must contain $O(q)$, for all $C^\prime$ s.t. $C\rsa C^\prime$, where `$\rsa$' means \emph{reaches in one or more steps}. Moreover, if all executions under consideration stabilize to an agreement, meaning that eventually all nodes stabilize to the same output, then the above implies that if an execution ever reaches a configuration containing $q$ then the output of that execution is necessarily $O(q)$.

A state $q$ is called \emph{reachable} if there is an initial configuration $C_0$ and an execution on $C_0$ that can produce $q$. We can also define reachability just in terms of the protocol, under the assumption that if $Q_0\subseteq Q$ is the set of initial states, then any possible combination of cardinalities of states from $Q_0$ can be part of an initial configuration. A \emph{production tree} for a state $q\in Q$, is a directed binary in-tree with its nodes labeled from $Q$ s.t. its root has label $q$, if $a$ is the label of an internal node (the root inclusive) and $b$, $c$ are the labels of its children, then the protocol has a rule of the form $\{b,c\}\rightarrow \{a,\cdot\}$ (that is, a rule producing $a$ by an interaction between a $b$ and a $c$ in any direction) \footnote{Whenever we use an unordered pair in a rule, like $\{b,c\}$, we mean that the property under consideration concerns both $(b,c)$ and $(c,b)$.}, and any leaf is labeled from $Q_0$. Observe now that if a path from a leaf to the root repeats a state $a$, then we can always replace the subtree of the highest appearance of $a$ by the subtree of the lowest appearance of $a$ on the path and still have a production tree for $q$. This implies that if $q$ has a production tree, then $q$ also has a production tree of depth at most $|Q|$, that is, a production tree having at most $2^{|Q|-1}$ leaves, which is a constant number, when compared to the population size $n$, that only depends on the protocol. Now, we can call a state $q$ \emph{reachable} (by a protocol $\ca$) if there is a production tree for it. These are summarized in the following proposition.

\begin{proposition} \label{pro:prod-tree}
Let $\ca$ be a protocol, $C_0$ be any (sufficiently large) initial configuration of $\ca$, and $q\in Q$ any state that is reachable from $C_0$. Then there is an initial configuration $C^\prime_0$ which is a sub-configuration of $C_0$ of size $n^\prime\leq 2^{|Q|-1}$ such that $q$ is reachable from $C^\prime_0$.
\end{proposition}
Proposition \ref{pro:prod-tree} is crucial for proving negative results, and will be invoked in Section \ref{sec:symmetry-breaking}.

\begin{proposition} \label{pro:mixed-output-stable}
Let $p$ be a predicate. There is no protocol that stably computes $p$ (all nodes eventually agreeing on the output in every fair execution), having both a reachable output-stable state with output $0$ and a reachable output-stable state with output $1$. 
\end{proposition}

An output-stable state $q$ is called \emph{disseminating} if $\{x,q\}\ra (q,q)$, for all $x\in Q$.

\begin{proposition} \label{pro:stable-disseminating}
Let $\ca$ be a protocol with at least one reachable output-stable state, that stably computes a predicate $p$ and let $Q_s\subseteq Q$ be the set of reachable output-stable states of $\ca$. Then there is a protocol $\ca^\prime$ with a reachable disseminating state that stably computes $p$.
\end{proposition}

\begin{theorem} \label{the:disseminating-symmetry}
Let $\ca$ be a protocol with a reachable disseminating state $q$ and let $\cc_0^d$ be the subset of its initial configurations that may produce $q$. Then the symmetry of $\ca$ on $\cc_0^d$ is $\Theta(N_{min})$.
\end{theorem}
Theorem \ref{the:disseminating-symmetry} emphasizes the fact that disseminating states can be exploited for maximum symmetry. We have omitted its proof, because it is similar to the proofs of Proposition \ref{pro:count-to-x} and Theorem \ref{the:general-positive}. This lower bound on symmetry immediately applies to single-signed linear combinations (where passing a threshold can safely result in the appearance of a disseminating state, because there are no opposite-signed numbers to inverse the process), thus, it can be used as an alternative way of arriving at Theorem \ref{the:general-positive}. On the other hand, the next proposition shows that this lower bound does not apply to linear combinations containing mixed signs, because protocols for them cannot have output-stable states.

\begin{proposition} \label{pro:mixed-signs-output-stable}
Let $p$ be a predicate of the form $\sum_{i\in [k]} a_iN_i\geq c$, for integer constants $k\geq 1$, $a_i$, and $c\geq 0$ such that at least two $a_i$s have opposite signs. Then there is no protocol, having a reachable output-stable state, that stably computes $p$. 
\end{proposition}

The proofs of Propositions \ref{pro:mixed-output-stable}, \ref{pro:stable-disseminating}, and \ref{pro:mixed-signs-output-stable} can be found in Appendix \ref{app:proofs-output-stable}.

\section{Harder Predicates}
\label{sec:symmetry-breaking}

In this section, we study the symmetry of predicates that, in contrast to single-signed linear combinations, do not allow for output-stable states. In particular, we focus on linear combinations containing mixed signs, like the \emph{majority} predicate, and also on modulo predicates like the \emph{parity} predicate. Recall that these predicates are not captured by the lower bound on symmetry of Theorem \ref{the:disseminating-symmetry}.

\subsection{Bounds for Mixed Coefficients}
\label{subsec:weak-positive}

We begin with a proposition stating that the majority predicate (also can be generalized to any predicate with mixed signs) can be computed with symmetry that depends on the difference of the state-cardinalities in the initial configuration. 

\begin{proposition} \label{pro:majority-weak-result}
The majority predicate $N_a-N_b>0$ can be computed with symmetry $\min\{N_{min},|N_a-N_b|\}$, where $N_{min}=\min\{N_a,N_b\}$. (The proof is given in Appendix \ref{app:majority-weak-result})
\end{proposition}

\begin{remark}
A result similar to Proposition \ref{pro:majority-weak-result} can be proved for any predicate of the form $\sum_{i\in [k]} a_iN_i-\sum_{j\in [h]} b_jN_j^\prime > c$, for integer constants $k,h,a_i,b_j\geq 1$ and $c\geq 0$.
\end{remark}

Still, as we prove in the following theorem, it is possible to do better in the worst case, and achieve any desired constant symmetry.

\begin{theorem} \label{the:majority-constant}
For every constant $k\geq 1$, the majority predicate $N_a-N_b>0$ can be computed with symmetry $k$. (The proof can be found in Appendix \ref{app:majority-weak-result})
\end{theorem}

\subsection{Predicates that Cannot be Computed with High Symmetry}
\label{subsec:negative}

We now prove a strong impossibility result, establishing that there are predicates that cannot be stably computed with much symmetry. The result concerns the \emph{parity} predicate, defined as $n\bmod 2 = 1$. In particular, all nodes obtain the same input, e.g., 1, and, thus, all begin from the same state, e.g., $q_1$. So, in this case, $N_{min}=n$ in every initial configuration, and we can here estimate symmetry as a function of $n$. The parity predicate is true iff the number of nodes is odd. So, whenever $n$ is odd, we want all nodes to eventually stabilize their outputs to $1$ and, whenever it is even, to $0$. If symmetry is not a constraint, then there is a simple protocol, described in \cite{AADFP04}, that solves the problem. Unfortunately, not only that particular strategy, but any possible strategy for the problem, cannot achieve symmetry more than a constant that depends on the size of the protocol, as we shall now prove. 

\begin{theorem} \label{the:parity-imp}
Let $\ca$ be a protocol with set of states $Q$, that solves the parity predicate. Then the symmetry of $\ca$ is less than $2^{|Q|-1}$.
\end{theorem}
\begin{proof}
For the sake of contradiction, assume $\ca$ solves parity with symmetry $f(n)\geq 2^{|Q|-1}$. Take any initial configuration $C_{n}$ for any sufficiently large odd $n$ (e.g., $n\geq f(n)$ or $n\geq |Q|\cdot f(n)$, or even larger if required by the protocol). By definition of symmetry, there is an execution $\alpha$ on $C_n$ that reaches stability without ever dropping the minimum cardinality of an existing state below $f(n)$. Call $C_{stable}$ the first output-stable configuration of $\alpha$. As $n$ is odd, $C_{stable}$ must satisfy that all nodes are in states giving output $1$ and that no execution on $C_{stable}$ can produce a state with output $0$. Moreover, due to the facts that $\ca$ has symmetry $f(n)$ and that $\alpha$ is an execution that achieves this symmetry, it must hold that every $q\in Q$ that appears in $C_{stable}$ has multiplicity $C_{stable}[q]\geq f(n)$.

Consider now the initial configuration $C_{2n}$, i.e., the unique initial configuration on $2n$ nodes. Observe that now the number of nodes is even, thus, the parity predicate evaluates to false and any fair execution of $\ca$ must stabilize to output 0. Partition $C_{2n}$ into two equal parts, each of size $n$. Observe that each of the two parts is equal to $C_{n}$. Consider now the following possible finite prefix $\beta$ of a fair execution on $C_{2n}$. The scheduler simulates in each of the two parts the previous execution $\alpha$ up to the point that it reaches the configuration $C_{stable}$. So, the prefix $\beta$ takes $C_{2n}$ to a configuration denoted by $2C_{stable}$ and consisting precisely of two copies of $C_{stable}$. Observe that $2C_{stable}$ and $C_{stable}$ consist of the same states with the only difference being that their multiplicity in $2C_{stable}$ is twice their multiplicity in $C_{stable}$. A crucial difference between $C_{stable}$ and $2C_{stable}$ is that the former is output-stable while the latter is not. In particular, any fair execution of $\ca$ on $2C_{stable}$ must produce a state $q_0$ with output $0$. But, by Proposition \ref{pro:prod-tree}, $q_0$ must also be reachable from a sub-configuration $C_{small}$ of $2C_{stable}$ of size at most $2^{|Q|-1}$. So, there is an execution $\gamma$ restricted on $C_{small}$ that produces $q_0$.

Observe now that $C_{small}$ is also a sub-configuration of $C_{stable}$. The reason in that (i) every state in $C_{small}$ is also a state that exists in $2C_{stable}$ and, thus, also a state that exists in $C_{stable}$ and (ii) the multiplicity of every state in $C_{small}$ is restricted by the size of $C_{small}$, which is at most $2^{|Q|-1}$, and every state in $C_{stable}$ has multiplicity at least $f(n)\geq 2^{|Q|-1}$, that is, $C_{stable}$ has sufficient capacity for every state in $C_{small}$. But this implies that if $\gamma$ is executed on the sub-configuration of $C_{stable}$ corresponding to $C_{small}$, then it must produce $q_0$, which contradicts the fact that $C_{stable}$ is output-stable with output 1. Therefore, we conclude that $\ca$ cannot have symmetry at least $f(n)\geq 2^{|Q|-1}$. 
\end{proof}

\begin{remark}
Theorem \ref{the:parity-imp} constrains the symmetry of any correct protocol for parity to be upper bounded by a constant that depends on the size of the protocol. Still, it does not exclude the possibility that parity is solvable with symmetry $k$, for any constant $k\geq 1$. The reason is that, for any constant $k\geq 1$, there might be a protocol with $|Q|>k$ (or even $|Q|\gg k$) that solves parity and achieves symmetry $k$, because  $k<2^{|Q|-1}$, which is the upper bound on symmetry proved by the theorem. On the other hand, the $2^{|Q|-1}$ upper bound of Theorem \ref{the:parity-imp} excludes any protocol that would solve parity with symmetry depending on $N_{min}$.
\end{remark}

\section{Further Research}
\label{sec:conclusions}

In this work, we managed to obtain a first partial characterization of the predicates with symmetry $\Theta(N_{min})$ and to exhibit a predicate (parity) that resists any non-constant symmetry. The obvious next goal is to arrive at an exact characterization of the allowable symmetry of all semilinear predicates. 

Another question, concerns the parity predicate, but could possibly apply to other modulo predicates as well. Some preliminary results of ours, indicate that constant symmetry for parity can be achieved if the initial configuration has a sufficient number of \emph{auxiliary nodes} in a distinct state $q_0$. It seems interesting to study how is symmetry affected by auxiliary nodes and whether they can be totally avoided.

Another very challenging direction for further research, concerns networked systems (either static or dynamic) in which the nodes have memory and possibly also unique IDs. Even though the IDs provide an \emph{a priori} maximum symmetry breaking, still, solving a task and avoiding the process of ``electing'' one of the nodes may be highly non-trivial. But in this case, defining the role of a process as its complete local state is inadequate. There are other plausible ways of defining the role of a process, but which one is best-tailored for such systems is still unclear and needs further investigation.

Finally, recall that in this work we focused on the \emph{inherent} symmetry of a protocol as opposed to its \emph{observed} symmetry. One way to study the observed symmetry would be to consider \emph{random parallel schedulers}, like the one that selects in every step a maximum matching uniformly at random from all such matchings. Then we may ask \emph{``What is the average symmetry achieved by a protocol under such a scheduler?''}. In some preliminary experimental results of ours, the expected observed symmetry of the \emph{Count-to-5} protocol (i) if counted until the alert state $q_5$ becomes an absolute majority in the population, seems to grow faster than $\sqrt{n}$ and (ii) if counted up to stability, seems to grow as fast as $\log n$ (see Appendix \ref{app:experiments} for more details).

\newpage

\appendix

\noindent\textbf{\huge APPENDIX}\normalsize 

\section{Proofs of Propositions \ref{pro:mixed-output-stable}, \ref{pro:stable-disseminating}, and \ref{pro:mixed-signs-output-stable}} \label{app:proofs-output-stable}

\noindent\textbf{Proposition \ref{pro:mixed-output-stable}.}
\emph{Let $p$ be a predicate. There is no protocol that stably computes $p$ (all nodes eventually agreeing on the output in every fair execution), having both a reachable output-stable state with output $0$ and a reachable output-stable state with output $1$.}

\begin{proof}
Let $\ca$ be such a protocol and let $q$ and $q^\prime$ be the two reachable output-stable states, such that $O(q)=0$ and $O(q^\prime)=1$. As both $q$ and $q^\prime$ are reachable, both have production trees. Denote by $L(a)$ the set of leaves of a minimum (w.r.t. its number of leaves) production tree for state $a$ together with their labels. Consider now an initial configuration $C_0$, in which a subset set of the nodes is labeled as $L(q)$ and another disjoint subset of nodes is labeled as $L(q^\prime)$. As any label in $L(q)\cup L(q^\prime)$ is an initial state, there is always such an initial configuration $C_0$. Now, consider the following finite prefix of a fair execution on $C_0$: The scheduler first simulates the production tree of $q$ on the nodes corresponding to $L(q)$ and then it simulates the production tree of $q^\prime$ on the nodes corresponding to $L(q^\prime)$. One way of performing the simulation is to always pick the next internal node of maximum depth, not picked yet. Such a node can only have nodes already picked or leaves as its children, otherwise there would be an unpicked internal node of greater depth, violating depth maximality. What the scheduler does after picking an internal node in state $a$ with children in states $b$ and $c$, is to select the appropriate interaction between two nodes in states $b$ and $c$ in order to produce a copy of state $a$ in the population. At the end of the first simulation, there will be a copy of state $q$ in the population and at the end of the second simulation there will be a copy of state $q^\prime$ in the population. As $q$ is output-stable, any fair execution having the above prefix producing $q$ must necessarily have eventually all nodes agree on output $O(q)=0$. But $q^\prime$ is also output-stable implying that all nodes must eventually agree on output $O(q^\prime)=1$, which is a contradiction.
\end{proof}

\noindent\textbf{Proposition \ref{pro:stable-disseminating}.} 
\emph{Let $\ca$ be a protocol with at least one reachable output-stable state, that stably computes a predicate $p$ and let $Q_s\subseteq Q$ be the set of reachable output-stable states of $\ca$. Then there is a protocol $\ca^\prime$ with a reachable disseminating state that stably computes $p$.}

\begin{proof}
We first show how to construct $\ca^\prime$ from $\ca$. Pick a single state $q\in Q_s$. Replace any occurrence of a $q^\prime\in Q_s$ in the transition function $\delta$ by $q$, eliminate duplicate rules, and remove from $Q$ all $q^\prime\in Q_s\bs\{q\}$. Finally, replace any rule $(x,y)\ra (z,w)$, where $x=q$, $y=q$, $z=q$, or $w=q$ by the rule $(x,y)\ra (w,w)$. This completes the construction of $\ca^\prime$. State $q$ is a disseminating state of $\ca^\prime$ because, by the last step of the construction, it holds that $\{x,q\}\ra (q,q)$ for all $x\in Q$. Moreover, $q$ is reachable because every $q^\prime\in Q_s$ is reachable in $\ca$, $q$ inclusive, and the above construction has only positively  affected the reachability of $q$.

It remains to show that $\ca^\prime$ stably computes $p$. As $\ca$ stably computes $p$, it suffices to show that when the two protocols are executed on the same schedule (including the choice of the initial configuration) their stable outputs are the same. Take any schedule in which $\ca$ produces a $q^\prime\in Q_s$ and consider the first time $t$ that this happens. As $q^\prime$ is output-stable, the stable output of $\ca$ on this schedule must be $O(q^\prime)$. Consider now $\ca^\prime$ on the same schedule. Before step $t$, the executions of $\ca^\prime$ and $\ca$ must be equivalent, because the construction has only affected rules containing at least one output-stable state. At step $t$, $\ca^\prime$ produces its disseminating state $q$, thus, its output is $O(q)=O(q^\prime)$, so the outputs of the two protocols agree on schedules producing output-stable states. Finally, for any schedule in which $\ca$ does not produce an output-stable state, the executions of $\ca^\prime$ and $\ca$ are equivalent on this schedule, thus, again their outputs agree.
\end{proof}

\noindent\textbf{Proposition \ref{pro:mixed-signs-output-stable}.} 
\emph{Let $p$ be a predicate of the form $\sum_{i\in [k]} a_iN_i\geq c$, for integer constants $k\geq 1$, $a_i$, and $c\geq 0$ such that at least two $a_i$s have opposite signs. Then there is no protocol, having a reachable output-stable state, that stably computes $p$.}

\begin{proof}
Let $\ca$ be such a protocol and let $q\in Q$ be a reachable output-stable state of $\ca$. Take an initial configuration $C_0$ that can produce $q$. As $q$ is output-stable, it must hold that the value of the predicate on $C_0$ is equal to $O(q)$. If $O(q)=1$ then $C_0$ must satisfy $\sum_{i\in [k]} a_iN_i\geq c$. Construct now another initial configuration $C_0^\prime$ by adding to $C_0$ as many nodes in a negative-coefficient initial state as required to violate $\sum_{i\in [k]} a_iN_i\geq c$. The value of the predicate $p$ on $C_0^\prime$ is equal to $0$ but $q$ can still be produced on the $C_0$ sub-configuration of $C_0^\prime$ implying that $\ca$'s output on $C_0^\prime$ is $1$. The latter violates the fact that $\ca$ stably computes $p$. If, instead, $O(q)=0$, then we can obtain a similar contradiction by adding a sufficient number of positive-coefficient nodes to $C_0$.
\end{proof}

\section{Proofs of Proposition \ref{pro:majority-weak-result} and Theorem \ref{the:majority-constant}}
\label{app:majority-weak-result}

\noindent\textbf{Proposition \ref{pro:majority-weak-result}.} 
\emph{The majority predicate $N_a-N_b>0$ can be computed with symmetry $\min\{N_{min},|N_a-N_b|\}$, where $N_{min}=\min\{N_a,N_b\}$.}

\begin{proof}
Initially, a node is in $(l,1)$ if its input is $a$ and in $(l,-1)$ if its input is $b$. The definition of the protocol is given in Protocol \ref{prot:majority}.

\floatname{algorithm}{Protocol}
\renewcommand{\algorithmiccomment}[1]{// #1}
\begin{algorithm}[!h]
  \caption{\emph{Majority}}\label{prot:majority}
  \begin{algorithmic}
    \medskip
    \State $Q=\{l,f\}\times\{-1,1\}$
    \State $I(a)=(l,1)$ and $I(b)=(l,-1)$
    \State $O(\cdot,-1)=0$ and $O(\cdot,1)=1$
    \State $\delta$: 
    \begin{align*}
    (l,i),(l,j)&\ra (f,-1),(f,-1), \mbox{ if } i+j=0\\
    (l,i),(f,j)&\ra (l,i),(f,i)\\
    (f,j),(l,i)&\ra (f,i),(l,i)\\
    (f,-1),(f,1)&\ra (f,-1),(f,-1)\\
    (f,1),(f,-1)&\ra (f,-1),(f,-1)
    \end{align*}
  \end{algorithmic}
\end{algorithm}

We first argue about the correctness of the protocol. Initially all nodes are $l$-leaders and $l$-leaders can only decrease via an interaction between an $(l,1)$ and an $(l,-1)$, in which case both become followers in state $(f,-1)$. The only things that followers do is to copy the data bit of the leaders (provided that at least one leader still exists) and to let the data bit $-1$ dominate a disagreement between two of them. Moreover, as long as there are at least two leaders with opposite data bits, due to fairness, an interaction between them will eventually occur. It follows that eventually, $\min\{N_a,N_b\}$ such eliminations will have occurred leaving $2\cdot\min\{N_a,N_b\}$ followers and $n-2\cdot\min\{N_a,N_b\}$ leaders. All leaders will have data bit 1 in case $a$s are the majority, data bit $-1$ in case $b$s are the majority, while there will be no leaders in case none of the two is a strict majority. In the first case, all followers will eventually copy 1, thus, all nodes will stabilize their output to 1. Observe that the $1$ of a follower may change many times to $-1$ due to its interactions with other followers that have not yet set their data bit to $1$, still fairness guarantees that eventually the unique (continuously reachable) stable configuration in which all followers have switched to $1$ after interacting with the leaders will occur. In the second case, all followers will eventually copy -1, thus, all nodes will stabilize their output to 0 and in the third case there are only followers, so the data bit -1 eventually dominates due to the last two rules of the protocol and eventually all nodes will stabilize their output to 0. In summary, if the $a$s form a strict majority all nodes stabilize to $1$, otherwise all nodes stabilize to $0$, thus, Protocol \ref{prot:majority} stably computes the majority predicate.

For symmetry, consider first those initial configurations which satisfy $|N_a-N_b|\geq \min\{N_a,N_b\}$. Consider the schedule that matches $\min\{N_a,N_b\}$ leaders with opposite data bits in its first step, leaving $|N_a-N_b|$ leaders agreeing on the majority (i.e., in the same state) and $2\cdot\min\{N_a,N_b\}$ followers in state $(f,-1)$. Up to this point, there is no symmetry breaking because the minimum cardinality that has appeared is still the initial minimum $\min\{N_a,N_b\}$. Next, the scheduler matches in one step $\min\{N_a,N_b\}$ followers to leaders and then in another step the rest $\min\{N_a,N_b\}$ followers to leaders, which leads to a stable configuration in which all followers have their output agree with the data bit of the leaders. As the minimum cardinality never has fallen below the initial minimum $\min\{N_a,N_b\}$, the symmetry is in this case at least $\min\{N_a,N_b\}$.  

Next, consider those initial configurations which satisfy $|N_a-N_b|<\min\{N_a,N_b\}$. Again, in the first step all opposite data bits are matched leaving $2\cdot\min\{N_a,N_b\}$ followers and $|N_a-N_b|\geq 0$ leaders. Observe that if $|N_a-N_b|=0$ then the configuration is already stable without any symmetry breaking. If $|N_a-N_b|\geq 1$, then the scheduler goes on by matching in one step $|N_a-N_b|$ followers to the leaders. Then it picks a leader and converts sequentially from the remaining followers, until precisely $|N_a-N_b|$ of them remain. Those are then converted in one step by being matched to the leaders. The minimum cardinality of a state in this schedule is $|N_a-N_b|$ and the initial minimum is $\min\{N_a,N_b\}$, so the symmetry breaking is $\min\{N_a,N_b\}-|N_a-N_b|$ and the symmetry is on those initial configurations $|N_a-N_b|$.
\end{proof}

\noindent\textbf{Theorem \ref{the:majority-constant}.} 
\emph{For every constant $k\geq 1$, the majority predicate $N_a-N_b>0$ can be computed with symmetry $k$.}

\begin{proof}
The idea is to multiply both $N_a$ and $N_b$ by $k$ so that their difference becomes $k|N_a-N_b|$. In this manner, the difference will become at least $k$ (in absolute value) whenever there is a strict majority which can be exploited for computation with symmetry $k$. Fortunately, multiplying both $N_a$ and $N_b$ by $k\geq 1$ does not affect the value of the majority predicate, but only the winning difference. 

To this end, we set the state of a node initially to $(l,k)$ if its input is $a$ and to $(l,-k)$ if its input is $b$. The definition of the protocol is given in Protocol \ref{prot:const-majority}.

\floatname{algorithm}{Protocol}
\renewcommand{\algorithmiccomment}[1]{// #1}
\begin{algorithm}[!h]
  \caption{\emph{$k$-Symmetry-Majority}}\label{prot:const-majority}
  \begin{algorithmic}
    \medskip
    \State $Q=\{l,f\}\times\{-k,-(k-1),\ldots,0,\ldots,k-1,k\}$
    \State $I(a)=(l,k)$ and $I(b)=(l,-k)$
    \State $O(\cdot,j)=0$, for all $-k\leq j\leq 0$ and $O(\cdot,i)=1$, for all $0<i\leq k$
    \State $\delta$: 
    \begin{align*}
    (l,i),(l,j)&\ra (l,i+j),(f,1), \mbox{ if } i,j \mbox{ have opposite signs and } i+j>0\\
               &\ra (l,i+j),(f,0), \mbox{ if } i,j \mbox{ have opposite signs and } i+j<0\\
               &\ra (f,0), (f,0), \mbox{ if } i+j=0\\
    (l,i),(f,\cdot)&\ra (l,i-1),(l,1), \mbox{ if } i\in\{2,3,\ldots,k\}\\
    (f,\cdot),(l,i)&\ra (l,1),(l,i-1), \mbox{ if } i\in\{2,3,\ldots,k\}\\    
    (l,j),(f,\cdot)&\ra (l,j+1),(l,-1), \mbox{ if } j\in\{-k,-(k-1),\ldots,-2\}\\
    (f,\cdot),(l,j)&\ra (l,-1),(l,j+1), \mbox{ if } j\in\{-k,-(k-1),\ldots,-2\}\\
    (l,i),(f,\cdot)&\ra (l,i),(f,i), \mbox{ if } i\in\{-1,0,1\}\\
    (f,\cdot),(l,i)&\ra (f,i),(l,i), \mbox{ if } i\in\{-1,0,1\}\\        
    (f,0),(f,1)&\ra (f,0),(f,0)\\
    (f,1),(f,0)&\ra (f,0),(f,0)
    \end{align*}
  \end{algorithmic}
\end{algorithm}

For correctness, initially all nodes are $l$-leaders. Now, $l$-leaders, apart from decreasing as in Protocol \ref{prot:majority}, can also increase. This occurs whenever an $(l,i)$, with $|i|\geq 2$, meets a follower, in which case the follower becomes a leader taking one unit of the other leader's count. Still, as in Protocol \ref{prot:majority}, as long as there are at least two leaders with opposite data bits, due to fairness, an interaction two such leaders will eventually occur. Eventually, all leaders will have a positive data bit in case $a$s are the majority, and a non-positive in case $a$s are not the majority. From that point on, no leader can change its output and all followers will eventually copy this output.

For symmetry, in case $N_a=N_b$, then the scheduler can pick a perfect bipartite matching between the $(l,k)$s and the $(l,-k)$s to convert them to $(f,0)$ and, thus, stabilize to output 0 without any symmetry breaking. The case, $N_b > N_a$ is simpler then the $N_a > N_b$ because the default output of the followers is 0, while in the $N_a > N_b$ case there is a small additional difficulty due to the fact that all $(f,0)$s have to be converted to $(f,1)$s. So, w.l.o.g. we focus on the $N_a>N_b$ case and we only give a proof for the special case in which $N_a=N_b+1$ as the other cases are similar. 

So, assume that $N_a=N_b+1$. The construction requires that $n\geq 2k(k+1)$. The initial configuration consists of $N_a=N_b+1$ nodes in state $(l,k)$ and $N_b$ nodes in $(l,-k)$. We present a schedule with symmetry $k$. The scheduler first picks a matching between $(l,k)$s and $(l,-k)$s of size $\lceil k/2\rceil$. This introduces $k$ copies of state $(f,0)$ and leaves $N_b-\lceil k/2\rceil$ nodes in state $(l,-k)$ and $N_b-\lceil k/2\rceil+1$ nodes in state $(l,k)$. Now isolate $k+1$ nodes in state $(l,k)$ and $k$ nodes in state $(l,-k)$. The remaining nodes in states $(l,k)$ and $(l,-k)$ (equal of each) are $n-k-(k+1)-k\geq 2k(k+1)-3k-1=2k^2-k-1$. These are converted to $(f,0)$s in one step, so we now have the initial $k$ $(f,0)$s, the new $(f,0)$s that are at least $2k^2-k-1$, $k+1$ $(l,k)$s, and $k$ $(l,-k)$s. Together the $(l,k)$s and $(l,-k)$s hold a total count of $k(k+1)+k^2=2k^2+k$ and together the new $(f,0)$s, the $(l,k)$s, and the $(l,-k)$s are at least $2k^2-k-1+(k+1)+k=2k^2+k$ nodes. So, there are enough nodes (the initial $(f,0)$s excluded) to distribute on them the count as follows. First the scheduler picks the $(l,k)$s and matches them in one step to $k+1$ nodes in $(f,0)$. This leaves $k+1$ nodes in $(l,k-1)$ and $k+1$ nodes in $(l,1)$. Then it matches the $(l,k-1)$s to $(f,0)$s, introducing $k+1$ more $(l,1)$s and leaving $k+1$ nodes in $(l,k-2)$. It continues in the same way until no count is greater than $1$, in this way distributing the counts of the $k+1$ nodes in $(l,k)$ to $k(k+1)$ copies of $(l,1)$. Observe that during this process the initial $(l,k)$s are always in identical states going in parallel through the sequence of states $(l,k),(l,k-1),(l,k-2),\ldots,(l,1)$, so each state on them is the state of $k+1$ other nodes. Moreover, their first matching with $(f,0)$s introduces $k+1$ $(l,1)$s in one step and from that point on (during this particular process) $(l,1)$s only increase, so the cardinality of $(l,1)$s does not go below $k+1$. Next (or in parallel), it does the same with the $k$ nodes in $(l,-k)$ leaving $k^2$ copies of $(l,-1)$. Observe that even though $(f,0)$s decrease during these processes, their cardinality never goes below $k$ due to the initial set of $k$ $(f,0)$s. So, at this point there are at least $k$ $(f,0)$s, $k^2+k$ $(l,1)$s, and $k^2$ $(l,-1)$, while the minimum multiplicity of a state has never dropped below $k$. The scheduler now matches in one step all the $(l,-1)$s to $(l,1)$s leaving in the population at least $k^2+k$ $(f,0)$s and precisely $k$ $(l,1)$s. Then, the scheduler matches all $(l,1)$s to $(f,0)$s, thus, introducing in one step $k$ $(f,1)$s (and still having at least $k^2$ $(f,0)$s), and then picks an $(l,1)$ and starts converting sequentially $(f,0)$s to $(f,1)$s until precisely $k$ $(f,0)$s have remained. Finally, it matches the remaining $k$ $(f,0)$s to the $k$ $(l,1)$s to convert all $(f,0)$s to $(f,1)$s in one step. At this point the protocol has stabilized and the multiplicity of no state has ever dropped below $k$.
\end{proof}

\section{Experiments on the Expected Observed Symmetry of Count-to-5}
\label{app:experiments}

\begin{figure}[!hbtp]
   \centering{
        \subfloat[]{
        \includegraphics[width=0.44\textwidth]{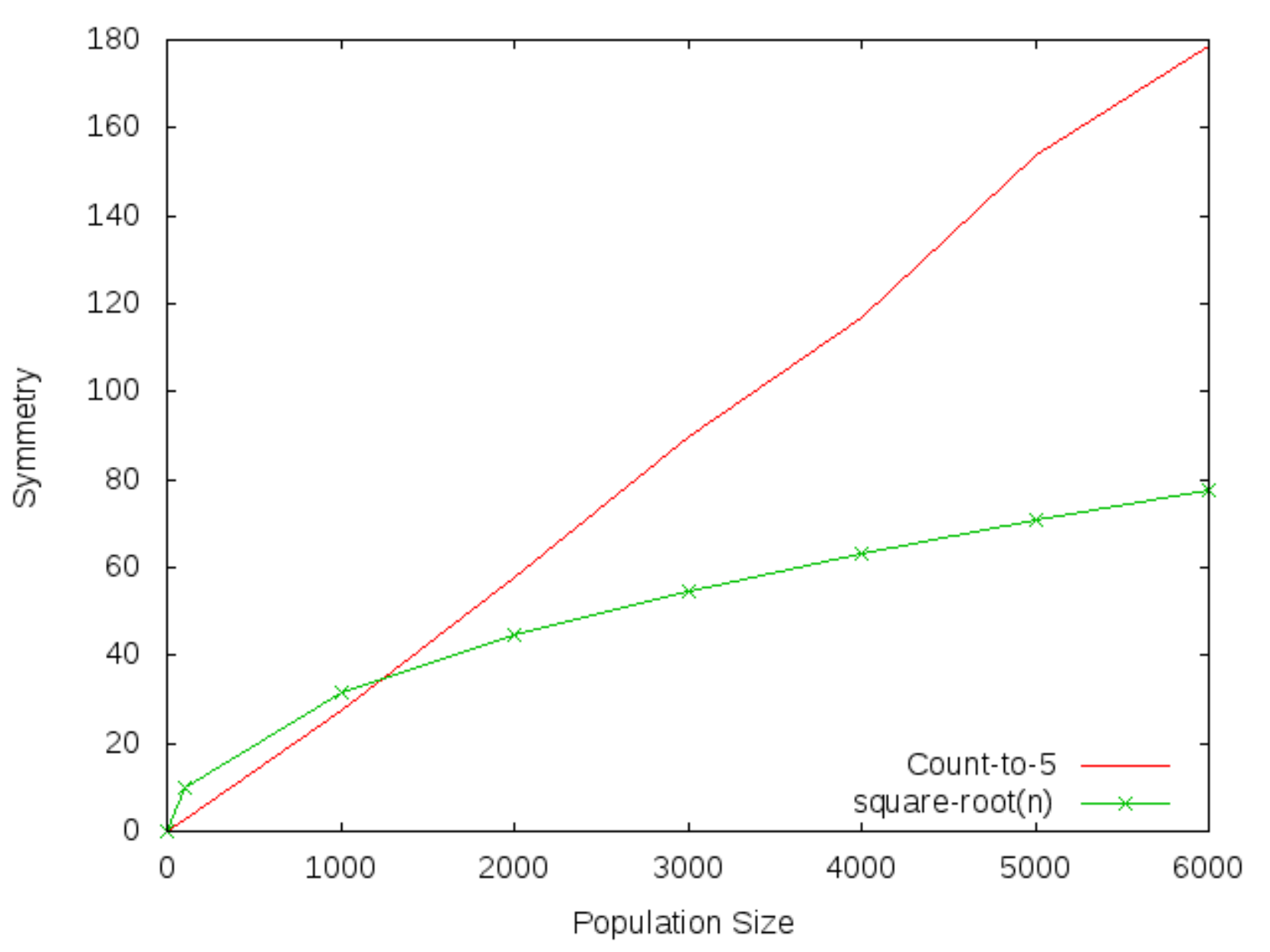}
        \label{fig:count-to-5-early}}
	\hspace{1cm}
        \subfloat[]{
        \includegraphics[width=0.44\textwidth]{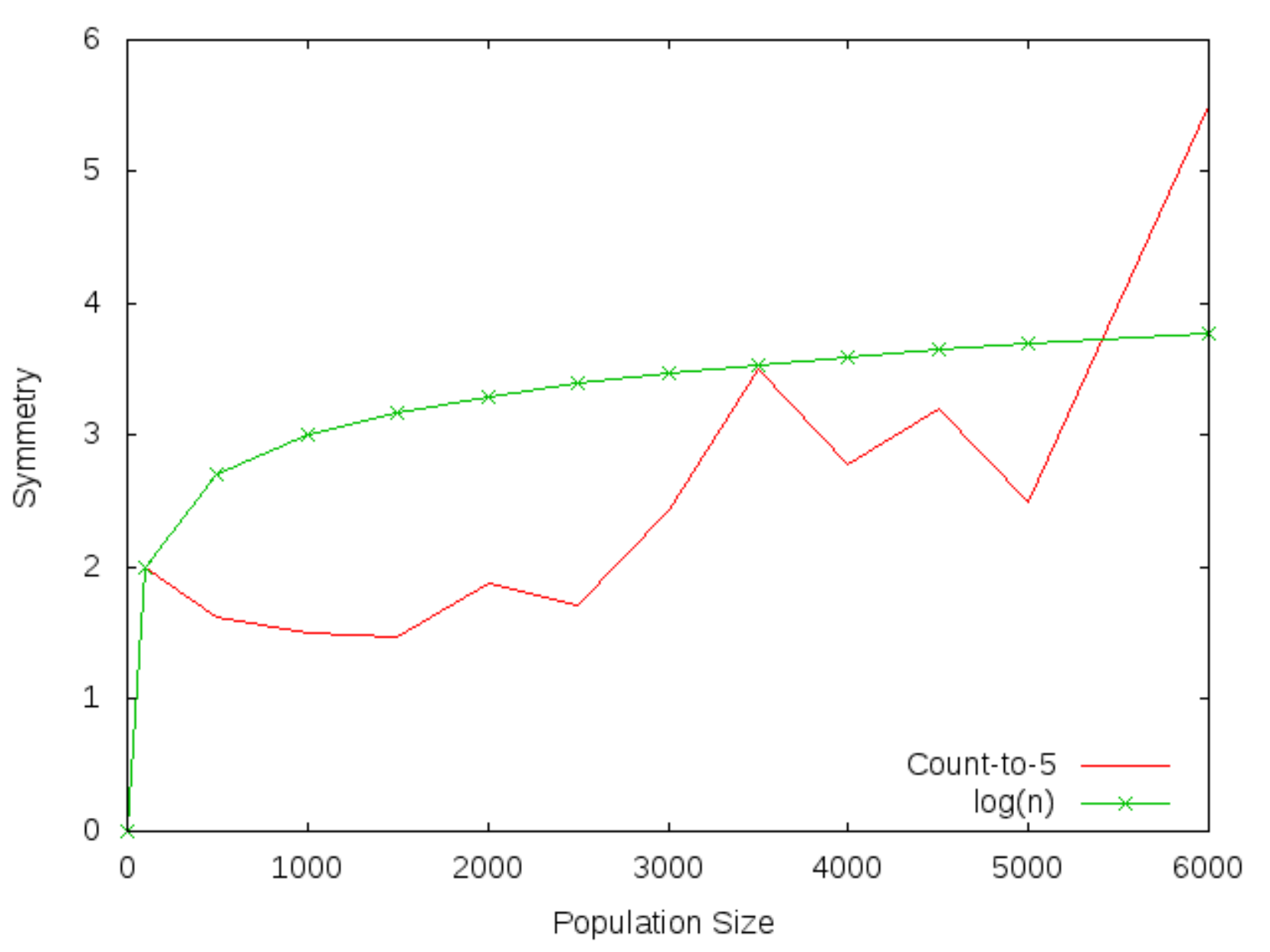}
        \label{fig:count-to-5-min}}
        }
   \caption{The experiments were performed with the NETCS simulator \cite{ALMS15}. The scheduler selects in every step a maximum cardinality matching, uniformly at random from all maximum matchings of the complete interaction graph. The implemented protocol is the \emph{Count-to-x} protocol of Section \ref{subsec:count-to-x}, for $x=5$. In a fraction of the populations of size up to 6000 nodes, several repetitions were performed and the average observed symmetry achieved by the protocol is plotted. The initial configuration is always the one resulting from assigning to all nodes the input value 1. (a) The observed symmetry of \emph{Count-to-5} (red line) is calculated up to the point that the alert state $q_5$ first becomes an absolute majority in the population and seems to grow faster than $\sqrt{n}$ (green line). (b) The observed symmetry of \emph{Count-to-5} (red line) is calculated up to stability and seems to grow as fast as $\log n$ (green line).} \label{fig:count-to-5-experiment}
\end{figure}

\noindent \textbf{Acknowledgements.} We would like to thank Dimitrios Amaxilatis for setting up and running, in the NETCS simulator of population protocols and network constructors \cite{ALMS15}, the experiments for the evaluation of the expected observed symmetry of the \emph{Count-to-5} protocol.

\end{document}